%% file: FraudFox.tex
\documentclass[runningheads]{llncs}
\usepackage{hyperref}
\usepackage{url}

\usepackage{xr}
\usepackage{algorithmic}
\usepackage{algorithm}
\usepackage{appendix}
\usepackage{epsfig}
\usepackage{epstopdf}
\usepackage{graphics}
\usepackage{graphicx}
\usepackage{lscape}
\usepackage{soul}
\usepackage{verbatim}
\usepackage{subcaption}

\usepackage{diagbox} 

\usepackage{amsfonts}
\usepackage{amssymb}
\usepackage{bbding}
\usepackage{bm}
\usepackage{booktabs} 
\usepackage{color}
\usepackage[section]{placeins}
\usepackage{graphicx}
\usepackage{natbib}
\usepackage{paralist}
\usepackage{xspace}

\input{dfn}

\begin{document}

\title{\method: Adaptable Fraud Detection in the Real World}

\hide{
\author{Matthew Butler}
\affiliation{%
  \institution{Amazon.com, Inc.}
  \streetaddress{410 Terry Ave N}
  \city{Seattle}
  \country{U.S.A.}}
\email{matbutle@amazon.com}

\author{Yi Fan}
\affiliation{%
  \institution{Amazon.com, Inc.}
  \streetaddress{410 Terry Ave N}
  \city{Seattle}
  \country{U.S.A.}}
\email{fnyi@amazon.com}

\author{Christos Faloutsos}
\affiliation{%
  \institution{Carnegie Mellon University}
  \streetaddress{5000 Forbes Avenue}
  \city{Pittsburgh}
  \country{U.S.A.}}
\email{faloutso@amazon.com}
}

\author{Matthew Butler\inst{1} \and
Yi Fan\inst{1} \and
Christos Faloutsos\inst{1,2}}
\authorrunning{M. Butler et al.}
%
\institute{Amazon.com, Inc., 410 Terry Ave N, Seattle WA 98109, U.S.A. 
\email{\{matbutle,fnyi\}@amazon.com}\\ \and
Carnegie Mellon University, 5000 Forbes Avenue, Pittsburgh, U.S.A.\\
\email{\{faloutso\}@amazon.com}}
\maketitle 

\begin{abstract}
\input{000abstract}

\end{abstract}

\hide{
}
\hide{
}



\section{Introduction}
\label{sec:introduction}
\input{010introduction}

\section{Literature Survey}
\label{sec:survey}

\input{020survey}

\section{Overview of proposed solution}
\label{sec:methods}

\input{030methods_approach}

\section{Proposed \methodW, for \oracleWeights}
\label{sec:EKF}
\input{031EKF_dynamicLinearModel}

\section{Proposed \methodI for \antiGaming}
\label{sec:antiGame}
\input{antiGame}

\section{Proposed \methodD for \decisionSurface}
\label{sec:FP_DS}
\input{041decisionSurfaceHyperbolic}

\section{Proposed \methodB for \businessConstraints}
\label{sec:paretoOptimal}
\input{paretoOptimal}

\section{Results}
\label{sec:Results}

\input{060results}

\section{Conclusions}
\label{sec:conclusion}
\input{080conclusion}

\paragraph*{Acknowlegements}
For her valuable contributions to the project we would like to thank Mina Loghavi.

\newpage

\bibliographystyle{plainnat}
\bibliography{BIB/citation}

\appendix

\input{091appendix_proofs}

\label{app:expForgetProof}

\input{increaseCovMatrix}

\label{app:increaseCovMatrix}

\input{002appendix_proof_decisionSurface}

\label{app:proofOptimalDecSur}

\end{document}

%% file: dfn.tex
\newtheorem{myProblem}{Problem}

\newtheorem{sfproblems*}{Semi-formal Problem}

\newtheorem{myLemma}{Lemma}

\newtheorem{myProof}{Proof}

\newcommand{\hide}[1]{}

\newcommand{\bit}{\begin{compactitem}}
\newcommand{\eit}{\end{compactitem}}
\newcommand{\ben}{\begin{compactenum}}
\newcommand{\een}{\end{compactenum}}
\newcommand{\beq}{\begin{equation}}
\newcommand{\eeq}{\end{equation}}

\newcommand{\nClassifiers}{k}
\newcommand{\ndim}{d}

\newcommand{\myvec}[1]{\vec{#1}}
\newcommand{\vecw}{\myvec{w}} 

\newcommand{\dsf}{F} 
\newcommand{\val}{v} 
\newcommand{\score}{s} 
\newcommand{\shMax}{\score_{H,max}} 
\newcommand{\invCost}{c} 
\newcommand{\pmargin}{m} 
\newcommand{\friction}{f} 
\newcommand{\hprof}{p_h} 
\newcommand{\fprof}{p_f} 

\newcommand{\oracleWeights}{{Fraud Adaptation}\xspace}
\newcommand{\decisionSurface}{{Decision Surface}\xspace}
\newcommand{\antiGaming}{{Anti-Gaming}\xspace}
\newcommand{\businessConstraints}{{Business Adaptation}\xspace}
\newcommand{\falseAlarm}{'false alarm'\xspace}

\newcommand{\principled}{{Principled\xspace}}
\newcommand{\scalable}{{Scalable\xspace}}
\newcommand{\effective}{{Effective/Deployed\xspace}}
\newcommand{\ensemble}{{Uses ensemble\xspace}}
\newcommand{\adversarial}{{Auto-adapts to adversaries\xspace}}
\newcommand{\nonstationary}{{Auto-adapts to non-stationary\xspace}}
\newcommand{\businessReq}{{Auto-adapts to business rec's\xspace}}
\newcommand{\myemph}[1]{{\bf #1}}

\newcommand{\automatic}{{Automatic\xspace}}
\newcommand{\method}{{\em FraudFox}\xspace}

\newcommand{\methodF}{{\method-{F}}\xspace}
\newcommand{\methodB}{{\method-{B}}\xspace}
\newcommand{\methodD}{{\method-{D}}\xspace}
\newcommand{\methodI}{{\method-{I}}\xspace}
\newcommand{\methodW}{{\methodF}}

%% file: 000abstract.tex
The proposed method (FraudFox) provides solutions to adversarial attacks in a resource constrained environment. We focus on questions like the following:
\newline
How suspicious is `Smith', trying to buy \$500 shoes, on Monday 3am?
How to merge the risk scores, from a handful of risk-assessment
modules (`oracles') in an adversarial environment?
More importantly, given historical data (orders, prices,
and what-happened afterwards),
and business goals/restrictions,
which transactions, like the `Smith' transaction above,
which ones should we `pass', versus send 
to human investigators?
The business restrictions could be: `at most $x$ investigations are feasible',
or `at most \$$y$ lost due to fraud'.
These are the two research problems we focus on,
in this work.
\newline
One approach to address the first problem (`oracle-weighting'),
is by using Extended Kalman Filters with dynamic importance weights, 
to
automatically and continuously update our weights
for each 'oracle'.
For the second problem, we show how to derive an optimal decision surface, 
and how to compute
the Pareto optimal set, to allow what-if questions.
An important consideration is adaptation:
Fraudsters will change their behavior, according to our
past decisions; thus, we need to adapt accordingly.
\newline
The resulting system, \method, is 
scalable, adaptable to changing fraudster behavior,
effective, and already in \textbf{production} at Amazon.
FraudFox augments a fraud prevention sub-system and has led to significant performance gains.

%% file: 010introduction.tex
Transactional fraud management is essential for ecommerce.
The resulting 'arms race' is the first issue 
that the proposed \method tackles.
The second issue, is non-adversarial changes:
the business rules/constraints do change over time,
e.g., next month we may want to increase (or reduce) the number of investigations.
The proposed \method handles this issue, too,
by carefully pre-computing a set of 'good' solutions.

\begin{figure}[htb]
	\centering
	\includegraphics[width=1.00\columnwidth]{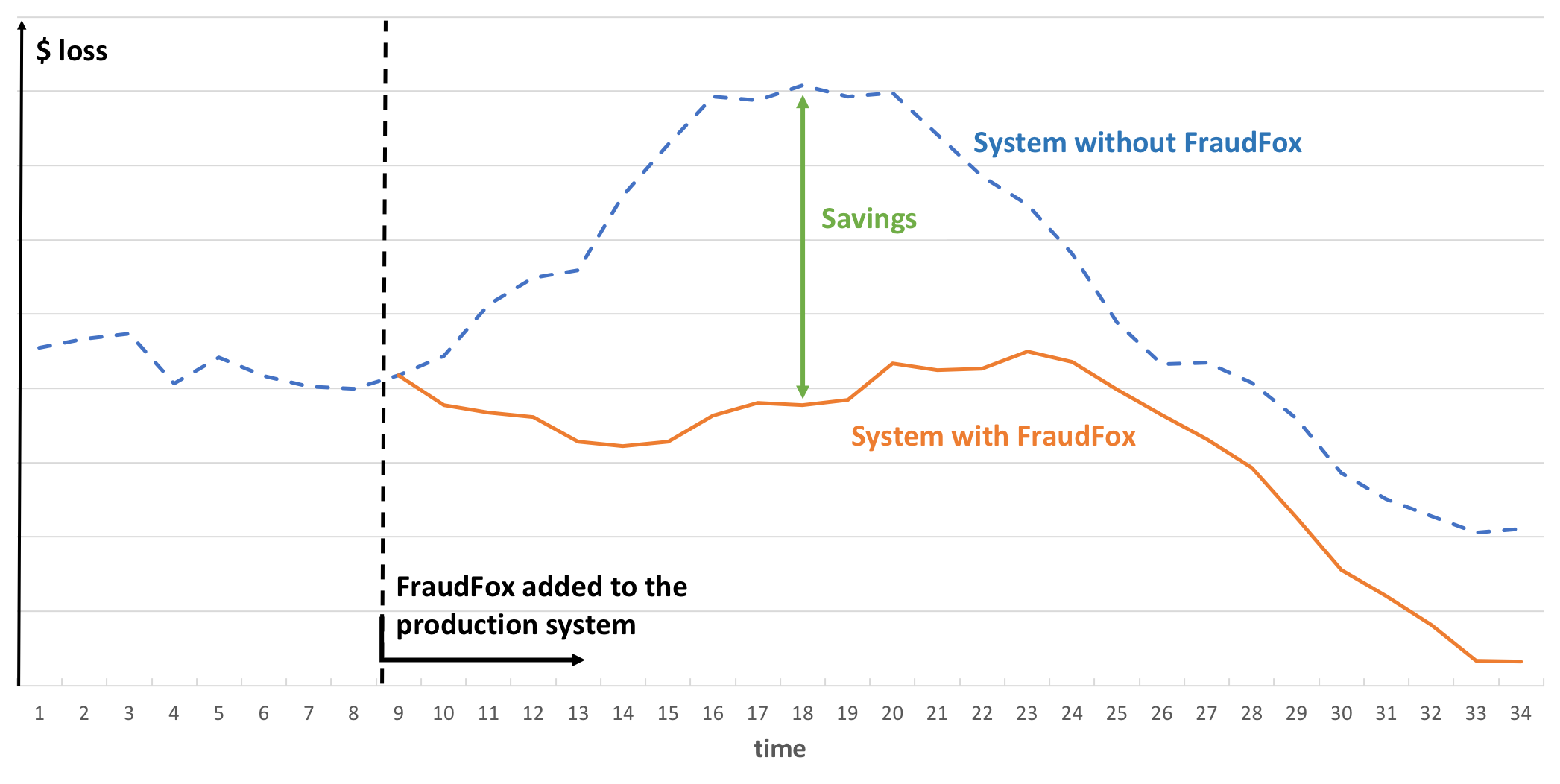}
	\caption[time-evolution]{{\bf Business Impact of \method}: 
		Visible loss reduction when applying it to real word data. 
		Losses vs. time - 1st without dynamically blending a subset of fraud indicators (blue), i.e., without \method; 2nd with dynamically blending (orange), i.e. with \method introduced at the dashed-black time tick.
		\label{fig:timeplot}}
\end{figure}

\par \noindent
Next we  address the specifications of the ADS track of the venue:

\begin{compactdesc}
	\item{\bf Deployed}: The system is operational (see Figure~\ref{fig:timeplot})
\item {\bf Real world challenge:}
    Spotting fraudsters is a major problem in e-retailers.
    The specific challenges we focus on, are the following two:
    (a) 	how to detect fraudsters that are behaving adversarially, i.e. that are carefully trying
        to avoid detection.
    (b) how to handle changing business constraints
        (like number of available human investigators).
\item {\bf Methodology - novelty}: For the first challenge, we propose an
   extended Kalman filter, with exponential decay 
   (see Theorem \ref{thm:EKFdecay} 
   on page \pageref{thm:EKFdecay}), so that
   it adapts to changing fraudster tactics.
   For the second challenge, we propose a decision surface 
   (see Lemma~\ref{lemma:hyperbola} on page~\pageref{lemma:hyperbola})
   and state-of-the-art optimization 
   with PSO (particle swarm optimization),
   to estimate the Pareto front
   (see section~\ref{sec:paretoOptimal}).
\item {\bf Business impact}:
   \method is in production,
   making a clear impact (see Figure~\ref{fig:timeplot}).
\end{compactdesc}

\par \noindent
We elaborate on the main problem and its subproblems.

At the high level, we want to make pass/investigate decisions automatically,
for each order like the 'Smith' order above, while
adapting to both
\bit
	\item the (adversarially) changing behavior 
	      of fraudsters and
	\item the (non-adversarially)
	      changing business constraints 
	      (eg., `$x$ more investigators are available, 
	      for the next $y$ months')
	\eit
\newpage
\par \noindent
A semi-formal problem definition is as follows:
\begin{sfproblems*}
	\label{ip:bp}
~ 
~
\begin{compactdesc}
\item [Given]
~
      \begin{itemize}
	       \item $\nClassifiers$ risk-assessing oracles
           \item historical data ( tuples of the form: [order-risk, order-price, importance-weight, fraud-flag] )
      \end{itemize}
\item [Find] 
~
      \begin{itemize}
      	    \item the oracle-weights $\vecw = (w_1, \ldots , w_\nClassifiers )$
      \end{itemize}
\item [To meet] 
~
      \begin{itemize}
      	\item the immediate  business goal/constraint, and 
      	\item to be prepared for unpredictable (or unplanned) business constraints in the near future.
      \end{itemize}
\end{compactdesc}
\end{sfproblems*}

As we describe later, we break the problems into four sub-problems, and we describe how we solve each. Figure~\ref{fig:timeplot} illustrates \method savings (green vertical arrow) post deployment (vertical black dashed line). The blue dashed line indicates fraud losses without dynamically blending a subset of fraud signals, while the orange solid line represents losses when performing this dynamic blending. Lower is better, and \method consistently achieves visible loss reduction. 

In short, our contributions are the following:
\begin{compactitem}
	\item \myemph{\automatic}: \method adapts to new elements of ground truth, without human intervention
	(see Theorem~\ref{thm:EKFdecay})
	\item \myemph{\principled}: 
	we derive from first principles an optimal decision surface (see Lemma~\ref{lemma:hyperbola})
	\item \myemph{\scalable}: \method is linear in the input size, and quadratic on
	the (small) number of oracles. It takes minutes to train on 
	a stock machine, and fractions of a second to make decisions
	\item \myemph{\effective}: Already in production, 
	\method has visible benefits (see Figure~\ref{fig:timeplot})
	
\end{compactitem}

The paper is organized in the usual way: survey, proposed solutions,
experiments, and conclusions.

%% file: 020survey.tex
\paragraph{Kalman filters}

To the best of our knowledge, there are few examples of Kalman Filters used in adversarial modeling. Notable work includes Park et. al~\cite{park2013intermittent} whom consider a control theory setting that involved a separated observer of a system and its controller. In this setting the adversary has the ability to non-randomly erase information transmitted by the observer to the controller. Their approach to counteracting the negative impacts of the adversary is based on non-uniform sampling of the observed system. This is not unlike our approach, where we non-uniformly weight the observation data based on a noisy signal from the adversary. Another example comes from Chang et al.~\cite{chang2018secure} where an ensemble method is proposed. The adversarial nature of the system is modeled using a ``secure estimator'' that accounts for the noise being introduced by the adversary not being zero mean i.i.d. Gaussian process noise. This secure estimator is then integrated into the canonical update equations of the KF. 

Additionally, our proposal overlaps with a few other areas of machine learning which will be discussed in the rest of the section. The canonical Kalman Filter \cite{kalman1960new} is a dynamic signal processor that was initially developed in the 60's. It has long been widely applied to signal processing, autonomous driving, target tracking and so on \cite{azuma1994SIGGRAPH}. Later, Extended Kalman Filter (\cite{sorenson1985IEEE},\cite{jazwinski2007KF},~\cite{niranjan1999sequential},~\cite{meinhold1983understanding}) and Unscented Kalman Filter \cite{julier2004IEEE} were proposed to solve a broader set of non-linear problems. For the Extended Kalman Filter, our work is similar to MacKay~\cite{mackay1992evidence} and Penny et al.~\cite{penny1999dynamic} where we make use of an approximation of the posterior from the former and also propose a method to deal with non-stationarity as in the later. However, their proposals do not consider applying the Kalman Filter to adversarial problems. 

\paragraph{Online learning}
The other key component of our proposal is online learning. Different online learning techniques have been proposed for linear (\cite{zhang2004ACM}, \cite{chen1994SIGMOD}) and non-linear (\cite{lakshminarayanan2014NIPS}, \cite{grnarova2018ICLR}) models, which are primarily designed to accelerate computation time and resolve scalability limitations. In the presence of adversarial attacks, online learning becomes more powerful due to its advantages in incorporating emerging data patterns in a more efficient and effective fashion, such as~\cite{quanrud2015NIPS} and the adversarial classification method from~\cite{dalvi2004KDD} (referred to as AC-method in table~\ref{tab:salesman}). However, no universal solution to these types of problems has been established due to the nature of adversaries.

\paragraph{Ensemble of classifiers}
Ensemble learning is a popular technique used to overcome individual shortcomings of one model or another. Most popular uses of ensembles come from boosting~\cite{freund1996experiments} and bootstrap-aggregation~\cite{breiman2001random} methods, as well as in a host of application studies~\cite{krawczyk2017ensemble}. Research in ensemble learning is still evolving, with some recent work focused on information extraction~\cite{ye2018rapidscorer},~\cite{lucchese2017quickscorer} or interpretability~\cite{tolomei2017interpretable} of these ensembles. While other work focuses on classification improvement~\cite{narassiguin2017dynamic} and optimal policies for combining classifiers~\cite{cerqueira2018constructive}. This includes specific applications, such as protein function prediction~\cite{yu2012transductive}, profit estimation~\cite{an2016map},\cite{xu2017ensemble} and time-series forecasting~\cite{saadallahdrift}\cite{cerqueira2017arbitrated}. Examples, of the use of Kalman Filters for ensemble learning include~\cite{peel2008data} and~\cite{baraldi2012kalman}, however neither application concerns adversarial or non-stationary environments. An example of ensemble learning in non-stationary environments includes~\cite{muhlbaier2007ensemble} where a simple weighted-majority vote is used to combine signals. However, the weights themselves are not learned in a principled way but are simply arbitrary weighted averages of the individual classifier performance over some time.

However, none of the above methods provides a full solution
to Problem~\ref{ip:bp} in the introduction.
The qualitative comparison is in Table~\ref{tab:salesman}
\begin{table}[htbp]
	\begin{center}
		\begin{tabular}{ l| c | c|c||c|}
			Property    & \rotatebox{80}{$L^{++}$~\cite{muhlbaier2007ensemble}} & \rotatebox{80}{AC~\cite{dalvi2004KDD}} & \rotatebox{80}{DLM~\cite{penny1999dynamic}} & \rotatebox{80}{\method} \\
			\hline  
			\adversarial & \hide{\XSolidBrush}& \CheckmarkBold & \hide{\XSolidBrush}&\CheckmarkBold \\
			\nonstationary &\CheckmarkBold & \hide{\XSolidBrush}& \CheckmarkBold&\CheckmarkBold \\
			\businessReq & \hide{\XSolidBrush}&\hide{\XSolidBrush} &\hide{\XSolidBrush} &\CheckmarkBold \\
			\effective    & \hide{\CheckmarkBold} & \hide{\CheckmarkBold} & \hide{\CheckmarkBold} &\CheckmarkBold \\
			\ensemble & \CheckmarkBold & \hide{\XSolidBrush} & \CheckmarkBold &\CheckmarkBold \\
			\hline \\
		\end{tabular}
		\caption{ \underline{\bf \method matches all specs}, while competitors
			miss one or more of the features.\label{tab:salesman}}
	\end{center}
\end{table}

%% file: 030methods_approach.tex
 We propose to break our high-level problems, into the following four sub-problems.
 The first two are for the adaptation to fraudsters
 (with adversarial behavior);
 the last two are for the adaptation to changing
 of business constraints.
 
 \begin{enumerate}
 	\item P\ref{prob:kalman} - \oracleWeights: find how to automatically update the weights of the $\nClassifiers$ classifiers, 
 	given new ground truth (see Problem~\ref{prob:kalman} page~\pageref{prob:kalman}).
 	Notice that we can give importance weights to each labeled order,
 	according to its  recency, or other criteria (see
 	page~\pageref{prob:antiGaming})
 	\item P\ref{prob:antiGaming} - \antiGaming: 
 	given the non-stationary and adversarial environment, determine the importance weights for new ground truth orders
 	(see Problem~\ref{prob:kalman} page~\pageref{prob:kalman})
 	\item P\ref{prob:surface} - \decisionSurface: find an optimal decision surface for passing or investigating orders, given
 	the various costs and benefits of making a decision (see page~\pageref{prob:surface})
 	\item P\ref{prob:businessConstraints} - \businessConstraints: 
 	Given the business constraints and need for quick adaption to constraint changes, find the best set (i.e., Pareto optimal) 
 	of classifier weights from Problem~\ref{prob:kalman}. 
 	See Problem~ \ref{prob:businessConstraints}, page~\pageref{prob:businessConstraints}.
 \end{enumerate}

Next, we show how to solve  the  four sub-problems.

%% file: 031EKF_dynamicLinearModel.tex
Here we focus on the sub-problem \oracleWeights.
Table~\ref{tab:symbols} 
lists the symbols and their definitions.
The exact definition is as follows

\begin{table}[htbp]
	\centering
	\begin{tabular}{|c|l|}
		\hline 
		Symbols & Definitions \\ \hline \hline
		$\nClassifiers$&  count of Oracles\\ 
		\hline 
		$\vecw_t$& weights of the Oracles (k x 1 vector) at time $t$ \\ 
		\hline 
		$\mathbf{\Sigma}_t$& covariance matrix of the Oracles (k x k matrix) \\ 
		\hline 
		$\mathbf{K}_t$& Kalman Gain \\ 
		\hline 
		$\mathbf{H}_t$& Observation matrix (Eq~\ref{eq:expSig})\\ 
		\hline 
		$\beta$& hyper-parameter of weighting policy (Eq~\ref{eq:border})\\ 
		\hline
		$\alpha$& hyper-parameter of weighting policy (Eq~\ref{eq:border}) \\
		\hline 
		$D_i$ & order $i$'s distance to decision surface (eq~\ref{eq:border})\\ 
		\hline 
		$z$ & binary class label \\ 
		\hline 
		$y$ & model output \\ 
		\hline 
		$s_t^2$ & variance of the activation function\\ 
		\hline 
	\end{tabular} 
\caption{Symbols and definitions \label{tab:symbols}}
\end{table}

\begin{myProblem}[Fraud Adaptation]
	~
	\label{prob:kalman}
	\begin{compactdesc}
		\item \textbf{Given}: 			a new order {\bf with} class label
		\item 	\textbf{Update}: 	the weight of each classifier ($w_1$, ... $w_\nClassifiers$)
		\item \textbf{to optimize}: classification accuracy
	\end{compactdesc}
\end{myProblem}

Our approach is based on taking a Bayesian perspective of canonical Logistic Regression of the form:
\begin{equation} \label{logisticRegression}
Y_t = P(\hat{Z_t} =1|\textbf{w}) = g(\textbf{w}^T\textbf{x}_t),
~~~\mathrm{where} 
~g(a_t) = \frac{exp(a_t)}{1 + exp(a_t)}
\end{equation}
where the predicted class label $Z_t$ is generated using the relationship defined in equation~\ref{logisticRegression} and $g(a_t)$ is the logistic function
and $a_t$ is the activation at time $t$. The activation being the linear combination of the model inputs ($\textbf{x}$) and the weights ($\textbf{w}$) of the model.

With the above setting, and with a new element of ground truth coming in
(an order with vector $\textbf{x}_t$, and fraud/honest flag $y_t$),
 the update equations for the weights of the 'oracles' $\hat{\textbf{w}} = [ w_1, ... , w_k]$ are given by
 the Extended Kalman Filter (EKF) equations:
 \begin{myLemma}[EKF update]
 	\begin{eqnarray}
 	\label{eq:mu}
 	\mathbf{\Sigma_t} &= & \mathbf{\Sigma}_{t-1} - \frac{y_t(1-y_t)}{1+y_t(1-y_t)s^2_t} (\mathbf{\Sigma}_{t-1}\textbf{x}_{t-1})(\mathbf{\Sigma}_{t-1}\textbf{x}_{t-1})^T \\
 	\label{eq:sigma}
 	\hat{\textbf{w}}_t &=  & \textbf{w}_{t-1} + \frac{\mathbf{\Sigma_t}}{1+y_t(1-y_t)s_t^2}\textbf{x}_t(z_t-y_t)
 	\end{eqnarray}
 \end{myLemma}
\begin{proof}
	Special case of the upcoming Theorem~\ref{thm:EKFdecay}, see Lemma~\ref{lem:noDecay}. \hfill $\blacksquare$
\end{proof}

%% file: antiGame.tex
Here we focus on the sub-problem of updating a classifier in a non-stationary and adversarial environment. The EKF facilities the incremental learning component, however equally important, is how much influence is awarded to new observations. The fraud environment is highly non-stationary and providing more weight to new observations is similar to gradually forgetting older ones. When receiving new fraud orders we know that all observations are not equally useful. Some of the orders are not actually fraud or the fraud pattern has already subsided. Given the large volume of new observations each day, it is impossible to manually make these decisions. 

Formally, the problem we propose and solve, is the following:

\begin{myProblem}[\antiGaming]
	\label{prob:antiGaming}
~

	~
\begin{compactdesc}
	\item \textbf{Given}:
	\begin{compactitem}
		\item Several labeled orders ($\ndim$-dim vectors) generated from a non-stationary and adversarial environment,
		\item and the decision surface (from problem~\ref{prob:surface}, see section~\ref{sec:paretoOptimal})
	\end{compactitem}
	\item \textbf{Find}: an appropriate weight for each order
	\item \textbf{To}: maximize classification accuracy on the recent time-period

\end{compactdesc}

\end{myProblem}

In this section we propose two approaches to assign importance weights. Firstly, we propose a novel weighting schema which considers fraudsters attempting to game the system (adversarial). Secondly, we derive novel update equations for the Kalman Filter that  accommodate non-stationary observations.

\subsection{Adversarial Adaptation}
    \label{sub:wp}

\input{032EKF_weightPolicy}

\subsection{Non-stationary Adaptation}
   \label{sub:exp}
   \input{033EKF_forgetting}

%% file: 032EKF_weightPolicy.tex
Building on the idea of an adversarial domain, it is observed that those looking to commit fraud will ``poke'' around a system looking for holes. In doing so, a typical fraud pattern can emerge that exists close to a $\textit{decision surface}$ of the fraud system. Figure~\ref{fig:gaming} illustrates the motivation behind our Eq.~\ref{eq:border}:
Fraudsters will tend to 'fly below the horizon', flocking below our decision surface. 
Our proposed approach will block them. A decision surface is essentially a rule which determines when an order is passed or investigated.  
We assume that the weight is a function of the distance from the decision surface, where observations closer to the decision surface receive more weight. An optimal weight is difficult to know a priori. Thus, we model this weight-to-distance relationship based on an exponential distribution. The weight ($\gamma$) of observation $i$ is given as:
\begin{equation}
\label{eq:border}
\gamma_i = \beta e ^{-\alpha D_i}
\end{equation}
where $\alpha$ and $\beta$ act as hyper-parameters of the model. A very practical aspect of the EKF is that it is parameter-less. However, this can lead to a lack of flexibility and so the weighting policy provides some opportunities to adapt the model. In a following section (\ref{sec:FP_DS}), we describe the creation of a decision surface from a cost-benefit analysis perspective, where standard approaches to calculate the distance a point, $i$, to a curve can be used for $D_i$.

\begin{figure}
	\centering
		\includegraphics[width=0.75\linewidth]{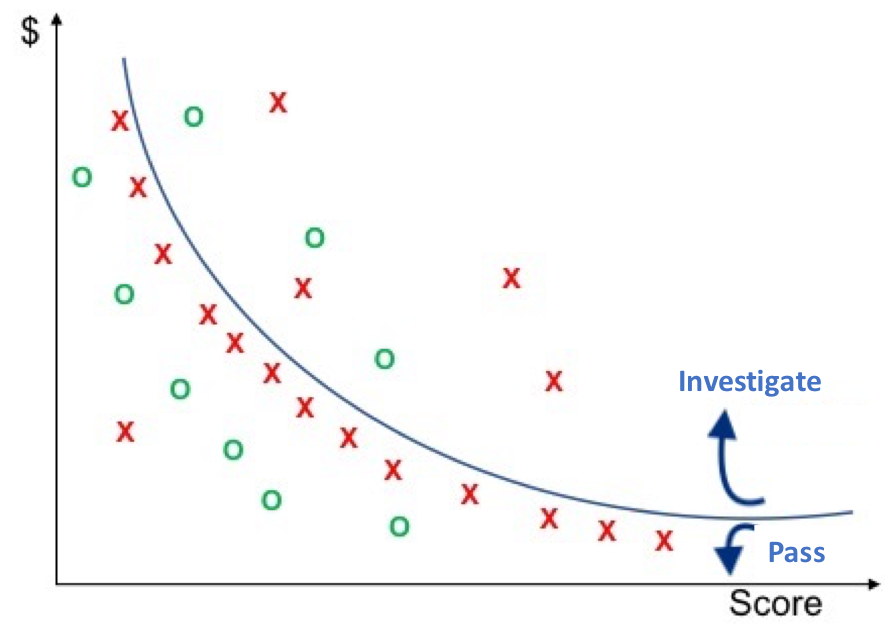}
	\caption[Gaming illustration]
	{ \underline{\bf Illustration of Decision Surface}
		and Adversarial Behavior.
		In the (fraud-score vs. order value) plot, 
		when we fix the decision surface (blue curve), 
		fraudulent orders (red crosses) will probably flock just below it.
		 \methodI is blocking exactly this behavior.
	\label{fig:gaming}}
\end{figure}

The values of ($\alpha$ \& $\beta$) are arbitrary, however, since training is very efficient we don't need to choose them a priori but can perform a simple grid search to generate a set of possible updates to the model.

%% file: 033EKF_forgetting.tex
In a non-stationary environment there is cause to gradually forget previously learned patterns. With online learning, this is achieved  with a forgetting function that diminishes model parameters overtime. With respect to \method, the ``idea'' of forgetting can be interpreted as an increase in the covariance matrix $\Sigma$, as a larger covariance entails more uncertainty in the model and thus less reliance of the past observations. When there is more uncertainty in the prior, the model updates will place more weight on the latest observations. An EKF model can be extended to non-stationary environments by simply including a term $q_t$ in the update equations, where $q_t\textbf{I}$ is isotropic covariance matrix describing the impact of the state noise on the current estimate of variance in the various priors. However, we see two problems with this approach: (1) It is difficult to accurately estimate the state-noise and (2) summarize it in $q_t$.

Rather than estimating the state noise, we approach the problem by estimating a weight (a function of time since the order was placed) to be applied to the current observation and thus the Kalman gain update. The full derivation of the equations below is included in the supplementary material. Formally, we consider a weight $a$ to be applied to the current Kalman Gain for purpose of a model update. Applying this weight requires the following:
\begin{theorem}[EKF with forgetting]
		\label{thm:EKFdecay}
With forgetting parameter of $a$, the update equations are:
\begin{equation}\label{eq:expMu}
\textbf{w}_t^* = \textbf{w}_{t-1} + a\boldsymbol{K_t}(z_t-y_t)  
\end{equation}
\begin{equation}\label{eq:expSig}
\boldsymbol{\Sigma_t} = \boldsymbol{\Sigma_{t-1}} + a(a-2)\boldsymbol{K}_t\boldsymbol{H}_t\boldsymbol{\Sigma_{t-1}}
\end{equation} 
where $\textbf{w}_t^*$ = $\boldsymbol{\mu}_t^*$ (see equation~\ref{eq:mean_update} in proof~\ref{app:proofExpForget} page~\pageref{app:proofExpForget}) and $\boldsymbol{H}_t$ is the derivative of the logit function. In our application $\boldsymbol{H}_t$ is a scalar, where:
\begin{equation}
\boldsymbol{H}_t = y_t(1-y_t)
\end{equation}
\end{theorem}
\begin{myProof}
	 See Appendix~\ref{app:proofExpForget}
	  \hfill $\blacksquare$
\end{myProof}
The value of $a$ can be chosen dynamically, and when $a$ is chosen to be $>$ 2 the covariance matrix increases in value with the update and thus forgets previous observations and relies more heavily on the current data point.
\begin{myLemma}
	\label{lem:noDecay}
	For $a$=1, i.e., no forgetting, Eq.~(\ref{eq:expMu}-\ref{eq:expSig}) become
	Eq.~(\ref{eq:mu}-\ref{eq:sigma})
\end{myLemma}
\begin{myProof}
	Recovering equations~\ref{eq:expMu}-\ref{eq:expSig} is achieved from straightforward substitution and reducing the update equations to their compact from, taking advantage of the scalar output. \hfill $\blacksquare$
\end{myProof}

In support of our argument that the added weight to new observations increases the uncertainty in the covariance matrix we provide the following proof. Where we show that the new update in equation~\ref{eq:expSig} guarantees that the resulting $\Sigma_t$ is always larger then or equal to the update from equation~\ref{eq:sigma}.
\begin{myLemma}
	\label{lem:incCovMat}
	For a given new observation $x_t$ with label $z_t$, we have:
	\begin{equation}
	\Sigma_t = \Sigma_{t-1} + a(a-2) K_tH_t\Sigma_{t-1} \ge \Sigma_{t-1} - K_tH_t\Sigma_{t-1}
	\end{equation}
\end{myLemma}
\begin{myProof}
	See Appendix~\ref{app:increaseCovMatrix}
	\hfill $\blacksquare$
\end{myProof}

%% file: 041decisionSurfaceHyperbolic.tex
In this section we derive the optimal decision surface for taking actions on orders, a necessity for our adversarial importance weighting.
For the purpose of this discussion, 
we simplify the actual system, 
and we assume that we need to choose between only two options,
namely
to pass, or to investigate.
Then, the problem is informally defined as follows:
\bit
\item given several orders, and their (price, score) for each,
\item decide which ones are best to pass (vs. investigate)
\eit

Formally, we have

\begin{myProblem}[\decisionSurface]
	~
	\label{prob:surface}
	\begin{compactdesc}
		\item \textbf{Given}:
		\begin{compactitem}
			\item the (dollar) cost of a \falseAlarm, false dismissal, and investigation cost
			\item and a specific order with suspiciousness score $\score$ and value (i.e., selling price) $\val$
		\end{compactitem}
		\item 	\textbf{Find} the best decision we should do (accept vs. investigate), 
		and expected value of this order, under our best decision.
	\end{compactdesc}
\end{myProblem}

In effect, we have to combine the price $\val$ and the score $\score$
(= probability of being fraudulent), into a function $F(s,u)$
to make our decision - if $F()$ is below a threshold $\theta$, 
we should accept the transaction as ``honest''; otherwise we should
investigate:
\beq
\mathrm{decision} = \left\{
    \begin{array}{l}
    \mathrm{`accept'~ if ~} F(s,u) \leq \theta \\
   \mathrm{`investigate' ~otherwise} 
   \end{array}
   \right.
\eeq
How should we choose the blending function $F()$, 
and the threshold $\theta$?
The straightforward (but wrong) answer, is to estimate
the expected value: $\dsf() = \val * \score$:
if the price $\val$ is low, and the probability of fraud
$\score$ is also low, then pass.
However, this is wrong: An expensive item
will almost always get investigated, potentially delaying a high-dollar order from a legitimate customer. 

The solution we propose is to  derive the result 
with a cost-benefit analysis, from {\em first principles}.
 Thus, we propose 
to also take into
account
\bit
\item $\invCost$: the \$ cost of the human investigation 
\item $\pmargin$: the profit margin (say, $m$=0.1, for 10\% profit margin)
\item $\friction$: the cost of friction (\$ amount reduction in life-time value or a coupon for a falsely investigated customer)
\item $d$: the loss, from a false dismissal (usually $d = 1-m$ but can also be impacted by external events, i.e. credit card declined by bank)
\eit

Figure~\ref{fig:costBenefit} illustrates our setting
\begin{figure}[htbp]
	\centering
\begin{tabular}{|c|c|c|}
	\hline 
	   \diagbox{guess}{Reality}      & Honest  ~ (1-$\score$) & Fraud ~ ($\score$) \\ 	\hline 
		Honest	&$\val * \pmargin$  & $ - \val * d$  \\ 	\hline 
		Fraud	&  $- \invCost - \friction$ &  $- \invCost$ \\ 	\hline 
\end{tabular} 
\caption{ \label{fig:costBenefit} Cost Benefit Analysis. Profit, for each
of the four cases.}
\end{figure}

Let $\shMax$ be the maximum fraud-score we should tolerate,
and still 'pass' this order. Then,
we can show that the resulting decision surface is 
a hyperbola (See illustration in Figure~\ref{fig:gaming}). 
Next, we prove this Lemma:
\begin{myLemma}[Hyperbolic decision surface]
	\label{lemma:hyperbola}
	The decision surface is a \textit{hyperbola}. Specifically,
	with the parameters as described above,
	the order  (price $\val$, prob. fraud $\score$) should be passed 
	(ie., treated as 'honest'):
	\[
		\mathrm{pass~ if~} \score < \shMax
	\]

	 where
	\begin{eqnarray}
	\shMax & = & 1 - \frac{\val  * d  - \invCost}{\val  * (\pmargin + d) + \friction} 
	\label{eq:decisionHonest}
	\end{eqnarray}
	\end{myLemma}
\begin{proof}
The idea is to find the expected profit for each of our decision,
and pick the highest. Thus
if we choose 'H' (honest), the expected profit $\hprof$ is
\[ \hprof = \val * \pmargin * (1-\score) + \score * (- \val) * d \]
That is,
with probability $\score$ we allow a fraudulent order,
and we lose its value ($- \val * d $);
and with probability (1-$\score$) we correctly allow the order,
for a profit of $\val * \pmargin$.
Similarly, the expected profit $\fprof$ when we decide 'fraud' is
\[ \fprof = -\invCost - \friction * (1-\score) ~ + ~  \score * (-\invCost)   \]
thus, we should decide to 'pass' (ie.,  'honest'), if
\[ 
\hprof > \fprof
\]
solving the inequality for $\score$, we complete the proof. See Appendix~\ref{app:proofOptimalDecSur} for the complete derivation. \hfill $\blacksquare$ 
\end{proof}

Notice that the resulting curve is a hyperbola in the 
($\score, \val$) plane.

We approximate this hyperbola as piece-wise linear to simplify the implementation in the production system. Among the many ways to fit our curve we chose a multi-objective meta-heuristic algorithm based on PSO~\cite{butler2012learning} due to a cheap objective function to compute and having two equally important objectives (\# investigated orders and fraud \$'s captured) to satisfy.

%% file: paretoOptimal.tex
\begin{myProblem}[\businessConstraints]
	~
	\label{prob:businessConstraints}
	\begin{compactdesc}
		\item \textbf{Given}:
		\begin{compactitem}
			\item the $m$ constraints of the business (like number of available human investigators)
			\item and the time-varying component of their thresholds
		\end{compactitem}
		\item 	\textbf{Find} an optimal set of solutions that satisfy the various thresholds of the $m$ business constraints.
	\end{compactdesc}
\end{myProblem}

We require a Pareto optimal set of solutions in business metric (BM) space. A solution on the Pareto frontier requires that one metric cannot be improved without the deterioration of another. Formally the Pareto set of solutions, P(Y), and a function, $f$, that maps a candidate \methodF model into BM space:
\begin{equation}
f: \mathbb{R}^n \rightarrow \mathbb{R}^m \nonumber
\end{equation}
where $X$ is a set of $n$ dimensional weight vectors found with the \methodF represented in the space $\mathbb{R}^n$, and $Y$ is the set of vectors in $\mathbb{R}^m$, our $m$ dimensional BM space, such that: 
\begin{equation}
Y = \{y \in \mathbb{R}^m: y = f(x), x \in X\}. \nonumber
\end{equation}
A point $y^*$ in $\mathbb{R}^m$ strictly dominates $y'$, represented as $y^*$ $\prec$ $y'$, when $y^*$ outperforms $y'$ in all $m$ dimensions. In figure~\ref{fig:gaming} (panel B) a non-dominated solution is shown as a point that it's NE quadrant (shown in light blue) is empty. Thus, the Pareto set is represented as:
\begin{equation}
P(Y) = \{y^* \in Y: \{y' \in Y : y^*\prec y',y^* \neq y'\} = \emptyset\}. \nonumber
\end{equation}

In our setting, the BM space can be represented as: (1) fraud dollars captured per investigation and (2) count of investigations performed.
The set of solutions is generated by using a grid search of the hyper-parameters ($\alpha$ and $\beta$) from equation~\ref{eq:border}, along with 
a varying ratio of positive-to-negative examples. 
In figure~\ref{fig:pareto} we plot model performance results after updating using the grid search represented in BM space, where the practitioner can choose a model that satisfies the current constraints. The models on the Pareto front in figure~\ref{fig:pareto} are indicated with a green 'X', dominated solutions are indicated with a red square.
\begin{figure}
	\centering
		\includegraphics[width=0.75\linewidth]{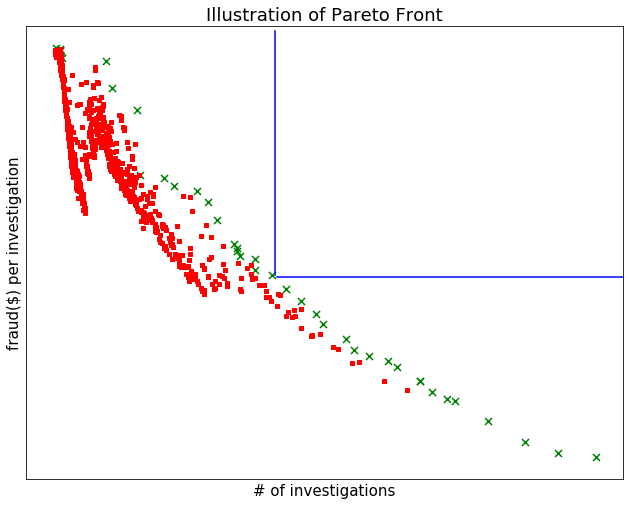}
	\caption[Pareto front illustration]
	{ \underline{\bf  \businessConstraints and Pareto Front.}
		 \method precomputes the Pareto front, in preparation of
		 changes in business requirements.
		 Scatter plot of model performance in business metric space (\# of investigations, vs  \$ fraud captured per investigation). Every point is a possible parameter choice of \method.
		 Points with empty
		  light blue NE quadrants are \textit{non-dominating} and thus form the Pareto front
		  (green 'x' points).
	\label{fig:pareto}}
\end{figure}

%% file: 060results.tex
\method has been implemented within a larger, live fraud prevention system.
In Figure~\ref{fig:timeplot}  (page~\pageref{fig:timeplot})
we provided a comparison (data points were smoothed via a moving average) of the missed fraud within the sub-system FraudFox was employed in,
with or without it.  
With the introduction of \method (dashed-black vertical line), there is a visible reduction in fraud losses.

\paragraph{Sanity-check: Adapting to shock}
We also present a synthetic example of a one-time change, see Figure~\ref{fig:simulation}.
In this case, \method adapts to the change, as expected.
In more detail,
figure~\ref{fig:simulation} shows a synthetic dataset with $k$=2 oracles, 
one near-perfect and the other near-random (and thus useless).
 At time $t$=30, there is a 'shock': the performance of the two
oracles swapped. Notice that the prediction loss surged immediately as the pattern changed. 
We show three different strategies: (1) {\em no-adaptation}: baseline - we keep the original
model ($\vecw$), unchanged; (2) {\em \method} (without exponential forgetting - section~\ref{sec:EKF};
and 
(3) {\method-exp}, with exponential forgetting of the older data points - subsection~\ref{sub:exp}.
The results in the Figure are exactly as expected:
 the loss jumps for  all methods at the moment of the 'shock',
 but  with {\em no-adaptation}, the loss stays high as expected;
while both versions of  \method adapt eventually, with the exponentially-forgetting one,
adapting faster.

\begin{figure}[tbp]
	\centering
	\includegraphics[width=0.95\linewidth]{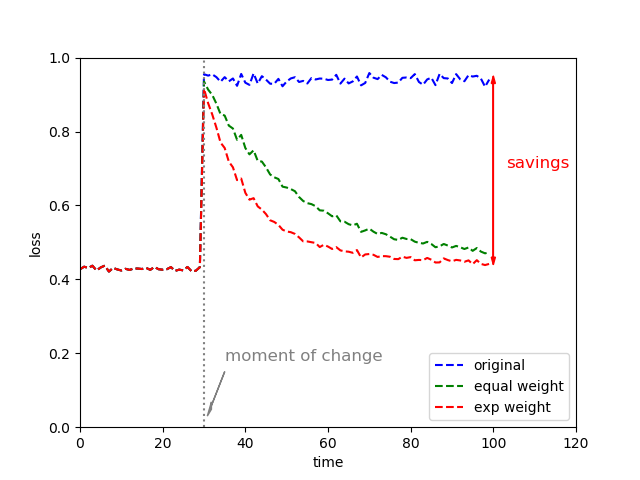}
	\caption[simulationI]{\underline{\method Adapts}: Synthetic Example. As expected,
		both our versions (equal-weights, in 'green' and  
		exponentially-decaying weights, in 'red') responds to a change/shock at $t$=30,
	    and eventually recover from the shock. Without adaptation ('blue'), 
        the behavior suffers.}
	\label{fig:simulation}
\end{figure}

%% file: 080conclusion.tex
In this work we have proposed a novel method (\method) 
that solves all the sub-problems listed in the introduction (see page~\pageref{sec:introduction}). \method automatically adapts to a dynamic environment and provides Pareto optimal alternatives to time-varying business constraints. 
Finally, \method is deployed in production and provides a solution, which is more resilient to abrupt changes in fraud behavior, while better balancing fraud losses and investigation costs.

In short, our main contributions are the following:

\begin{compactitem}
	
	\item \myemph{\automatic}: \method adapts to both adversarial and non-stationary observations
	(See Theorem~\ref{thm:EKFdecay} and Section~\ref{sec:paretoOptimal})
	\item \myemph{\principled}: 
	we derived from first principles
	an optimal decision surface from investigating orders based on a cost-benefit framework 
	(see Lemma~\ref{lemma:hyperbola})
	\item \myemph{\scalable}: \method is linear on the input size for parameter updates,
	and takes below a second, for decisions on the order-level.
	\item \myemph{\effective}:  \method is already in production,
	with  visible benefits.
	
\end{compactitem}

%% file: 091appendix_proofs.tex
\section{Appendix: Proof of Exponential Forgetting Formulae}
\label{app:proofExpForget}

Here we give the proof of Theorem~\ref{thm:EKFdecay} on page~\pageref{thm:EKFdecay}. We repeat the theorem, for convenience.

\noindent\textbf{Theorem 1[EKF with forgetting]}
With forgetting parameter of $a$, the update equations are:
\begin{equation}
\textbf{w}_t^* = \textbf{w}_{t-1} + a\boldsymbol{K_t}(z_t-y_t)  
\end{equation}
\begin{equation}
\boldsymbol{\Sigma_t} = \boldsymbol{\Sigma_{t-1}} + a(a-2)\boldsymbol{K}_t\boldsymbol{H}_t\boldsymbol{\Sigma_{t-1}}
\end{equation} 

\noindent\textbf{Proof} Let's first set up some notations.
\begin{eqnarray}
\boldsymbol{w}_t & = & \boldsymbol{w}_{t-1} + \boldsymbol{\epsilon}_t \\
z_t & =  &h(\boldsymbol{w}_t,\boldsymbol{x}_t) + \nu_t
\end{eqnarray}
The underlying assumption is $E(\boldsymbol{\epsilon}_t)=\boldsymbol{0}$, $E(\nu_t)=\boldsymbol{0}$, and $E(\boldsymbol{\epsilon}_t^T\nu_t)=0$. Furthermore, we denote $var(\boldsymbol{\epsilon}_t)=\boldsymbol{Q}_t$ and $var(\nu_t)=R_t$.\\

\noindent From regular EKF, we know that the update function of the mean is:
\begin{equation}
\boldsymbol{\mu}_t = \boldsymbol{\mu}_{t-1} + \boldsymbol{K_t}(z_t - y_t),
\end{equation}
where $\boldsymbol{K_t}$ is called Kalman Gain with $\boldsymbol{K_t} = \boldsymbol{\Sigma_{t-1}^\ast} \boldsymbol{H_t}(\boldsymbol{H_t}\boldsymbol{\Sigma_{t-1}^\ast} \boldsymbol{H_t^T} + R_t)^{-1}$, $\boldsymbol{H_t}$ is the partial derivative of $h(\boldsymbol{w}_t,\boldsymbol{x}_t)$ w.r.t. $\boldsymbol{w}_t$, and $y_t = h(\boldsymbol{\mu}_{t-1},\boldsymbol{x}_t)$.\\

\noindent Now, with the new proposed mean update function:
\begin{equation}\label{eq:mean_update}
\boldsymbol{\mu}_t^* = \boldsymbol{\mu}_{t-1} + a\boldsymbol{K_t}(z_t - y_t)
\end{equation}
we have the new error term as
\begin{eqnarray}
\tilde{\boldsymbol{e}}_t &= & \boldsymbol{w}_t - \boldsymbol{\mu}_t^\ast \nonumber\\
& = & \boldsymbol{w}_{t-1} + \boldsymbol{\epsilon}_t - (\boldsymbol{\mu}_{t-1} + a\boldsymbol{K_t}(z_t - y_t)) \nonumber\\
& = & (\boldsymbol{w}_{t-1} - \boldsymbol{\mu}_{t-1}) + \boldsymbol{\epsilon}_t - a\boldsymbol{K_t}(h(\boldsymbol{w}_t) + \nu_t - h(\boldsymbol{\mu}_{t-1})) \nonumber\\
& = & \boldsymbol{e}_{t-1} + \boldsymbol{\epsilon}_t - a\boldsymbol{K}_t\boldsymbol{H}_t(\boldsymbol{w}_t - \boldsymbol{\mu}_{t-1}) - a\boldsymbol{K}_t\nu_t \nonumber\\
& = & \boldsymbol{e}_{t-1} + \boldsymbol{\epsilon}_t - a\boldsymbol{K}_t\boldsymbol{H}_t(\boldsymbol{w}_t - \boldsymbol{w}_{t-1} + \boldsymbol{w}_{t-1} - \boldsymbol{\mu}_{t-1}) - a\boldsymbol{K}_t\nu_t \nonumber\\
& = & \boldsymbol{e}_{t-1} + \boldsymbol{\epsilon}_t - a\boldsymbol{K}_t\boldsymbol{H}_t(\boldsymbol{e}_{t-1} + \boldsymbol{\epsilon}_t) - a\boldsymbol{K}_t\nu_t \nonumber\\
& = & (\boldsymbol{I} - a\boldsymbol{K}_t\boldsymbol{H}_t)\boldsymbol{e}_{t-1} + (\boldsymbol{I} - a\boldsymbol{K}_t\boldsymbol{H}_t)\boldsymbol{\epsilon}_t - a\boldsymbol{K}_t\nu_t.\nonumber
\end{eqnarray}

\noindent Therefore,
\begin{eqnarray}
\boldsymbol{\Sigma}_t & = & E(\tilde{\boldsymbol{e}}_t\tilde{\boldsymbol{e}}_t^T) \nonumber\\
\nonumber
& =&  (\boldsymbol{I} - a\boldsymbol{K}_t\boldsymbol{H}_t)\boldsymbol{\Sigma}_{t-1}(\boldsymbol{I} - a\boldsymbol{K}_t\boldsymbol{H}_t)^T + \nonumber\\
& & ~~~ (\boldsymbol{I} - a\boldsymbol{K}_t\boldsymbol{H}_t)\boldsymbol{Q}_t(\boldsymbol{I} - a\boldsymbol{K}_t\boldsymbol{H}_t)^T + 
 a^2\boldsymbol{K}_tR_t\boldsymbol{K}_t^T \nonumber\\ 
& = & (\boldsymbol{I} - a\boldsymbol{K}_t\boldsymbol{H}_t)\boldsymbol{\Sigma}_{t-1}^\ast (\boldsymbol{I} - a\boldsymbol{K}_t\boldsymbol{H}_t)^T + a^2\boldsymbol{K}_tR_t\boldsymbol{K}_t^T \nonumber\\
\nonumber
& = & \boldsymbol{\Sigma}_{t-1}^\ast - a\boldsymbol{K}_t\boldsymbol{H}_t\boldsymbol{\Sigma}_{t-1}^\ast - a\boldsymbol{\Sigma}_{t-1}^\ast \boldsymbol{H}_t^T\boldsymbol{K}_t^T + \nonumber\\
& &  ~~~ a^2\boldsymbol{K}_t\boldsymbol{H}_t\boldsymbol{\Sigma}_{t-1}^\ast \boldsymbol{H}_t^T\boldsymbol{K}_t^T + a^2\boldsymbol{K}_tR_t\boldsymbol{K}_t^T \nonumber
\end{eqnarray}
where $\boldsymbol{\Sigma}_{t-1}^\ast = \boldsymbol{\Sigma}_{t-1} + \boldsymbol{Q}_t$. \nonumber\\

\noindent We use the same $\boldsymbol{K}_t$ as in the regular EKF and adapt that into the above equation to get:
\begin{eqnarray}
\boldsymbol{\Sigma}_t & = (\boldsymbol{I} - a\boldsymbol{K}_t\boldsymbol{H}_t)\boldsymbol{\Sigma}_{t-1}^\ast - a(\boldsymbol{I} - a\boldsymbol{K}_t\boldsymbol{H}_t)\boldsymbol{\Sigma}_{t-1}^\ast \boldsymbol{H}_t^T\boldsymbol{K}_t^T + a^2\boldsymbol{K}_tR_t\boldsymbol{K}_t^T \nonumber\\
& = (\boldsymbol{I} - a\boldsymbol{K}_t\boldsymbol{H}_t)\boldsymbol{\Sigma}_{t-1}^\ast - [\boldsymbol{\Sigma}_{t-1}^\ast \boldsymbol{H}_t^T - a\boldsymbol{K}_t(\boldsymbol{H}_t\boldsymbol{\Sigma}_{t-1}^\ast \boldsymbol{H}_t^T + R_t)]a\boldsymbol{K}_t^T \nonumber
\end{eqnarray}

Note that, $\boldsymbol{K}_t(\boldsymbol{H}_t\boldsymbol{\Sigma}_{t-1}^\ast \boldsymbol{H}_t^T + R_t) = \boldsymbol{\Sigma}_{t-1}^\ast \boldsymbol{H}_t$. Therefore,
\begin{eqnarray}
\boldsymbol{\Sigma}_t & = & (\boldsymbol{I} - a\boldsymbol{K}_t\boldsymbol{H}_t)\boldsymbol{\Sigma}_{t-1}^\ast - [\boldsymbol{\Sigma}_{t-1}^\ast \boldsymbol{H}_t^T - a\boldsymbol{\Sigma}_{t-1}^\ast \boldsymbol{H}_t]a\boldsymbol{K}_t^T \nonumber\\
& = & \boldsymbol{\Sigma}_{t-1}^\ast - a\boldsymbol{K}_t\boldsymbol{H}_t\boldsymbol{\Sigma}_{t-1}^\ast - a\boldsymbol{\Sigma}_{t-1}^\ast \boldsymbol{H}_t^T\boldsymbol{K}_t^T + a^2\boldsymbol{\Sigma}_{t-1}^\ast \boldsymbol{H}_t\boldsymbol{K}_t^T \nonumber\\
& = & \boldsymbol{\Sigma}_{t-1}^\ast + a(a-2)\boldsymbol{K}_t\boldsymbol{H}_t\boldsymbol{\Sigma}_{t-1}^\ast. \nonumber
\end{eqnarray}

\noindent If we let $\boldsymbol{Q} \rightarrow 0$, then $\boldsymbol{\Sigma}_{t-1}^\ast \rightarrow \boldsymbol{\Sigma}_{t-1}$ and we have $\boldsymbol{\Sigma}_t = \boldsymbol{\Sigma}_{t-1} + a(a-2)\boldsymbol{K}_t\boldsymbol{H}_t\boldsymbol{\Sigma}_{t-1}$.
 \hfill $\blacksquare$

%% file: increaseCovMatrix.tex
\section{Appendix: Proof of Larger Covariance Matrix}

In support of our argument that the added weight to new observations increases the uncertainty in the covariance matrix we provide the following proof. Where we show that the new update in equation~\ref{eq:expSig} guarantees that the resulting $\Sigma_t$ is always larger then or equal the update from equation~\ref{eq:sigma}.

\noindent\textbf{Lemma 3[Increase in Covariance Matrix]}

For $\Sigma_t$ = $\Sigma_{t-1}$ + a(a-2) $K_tH_t\Sigma_{t-1}$ $\ge$ $\Sigma_{t-1} - K_tH_t\Sigma_{t-1}$

\noindent\textbf{Proof} 

The equality above holds when a = 1. Note, the right-hand side is the covariance update of the original EKF (Extended Kalman Filter).
To see this,
\begin{eqnarray}
&&\Sigma_{t-1}+a(a-2) K_t H_t \Sigma_{t-1}-( \Sigma_{t-1}- K_t H_t \Sigma_{t-1} ) \nonumber\\
&=& (a^2-2a+1)K_t H_t \Sigma_{t-1} \nonumber\\ 
&=& (a-1)^2 K_t H_t \Sigma_{t-1} \label{aInc}
\end{eqnarray}

Since $K_tH_t\Sigma_{t-1}$ is positive definite and symmetric, we can multiply the above by 
($K_tH_t\Sigma_{t-1})^{-1}$ and the formula is reduced to $(a-1)^2 \ge $0, where the equality holds when a=1.

More rigorously, we can prove the invertibility in the following: 
Let M= $K_tH_t\Sigma_{t-1}$. Since $M$ is symmetric, we decompose $M$ as $QDQ^T$. Here, $D$ is a diagonal matrix of eigenvalues of $M$ and $Q$ is an orthogonal matrix, where rows of $Q$ are eigenvectors of $M$. Hence, equation~\ref{aInc} is equal to $(a-1)^2$ $QDQ^T$.

Since $Q$ is by definition invertible, we can multiply $Q^T$ to the left and $(Q^{T)^{-1}}$ to the right and reduce it further to $(a-1)^2 D$.

Now $D$ contains the eigenvalues from a symmetric covariance matrix and therefore all eigenvalues are $\ge$ 0, thus $(a-1)^2$ $D$$>$0 for any $a$$>$1. The canonical form with $a$=1 will always lead to min variance for a given update.
 \hfill $\blacksquare$

%% file: 002appendix_proof_decisionSurface.tex
\section{Appendix: Proof of Optimal Decision Surface for Guessing Fraud}
\label{app:proofOptimalDecSur}

Here we provide the full proof of Lemma~\ref{lemma:hyperbola} on page~\pageref{lemma:hyperbola}. We repeat the lemma, for convenience, as well as the definitions: 
\bit
\item $\invCost$: the \$ cost of the human investigation 
\item $\pmargin$: the profit margin (say, $m$=0.1, for 10\% profit margin)
\item $\friction$: the cost of friction (\$ amount reduction in life-time value or a coupon for a falsely investigated customer)
\item $d$: the loss, from a false dismissal (usually $d = 1-m$ but can also be impacted by external events, i.e. credit card declined by bank)
     
\eit

\noindent\textbf{Lemma 4 [Optimal Decision Surface]}
	The decision surface is a \textit{hyperbola}. Specifically,
	with the parameters as described above,
	the order  (price $\val$, prob. fraud $\score$) should be passed 
	(ie., treated as 'honest'):
	\[
		\mathrm{pass~ if~} \score < \shMax
	\]

	 where
	\begin{eqnarray}
	\shMax & = & 1 - \frac{\val  * d  - \invCost}{\val  * (\pmargin + d) + \friction} 
	\label{eq:decisionHonest}
	\end{eqnarray}

\noindent\textbf{Proof} 

The idea is to find the expected profit for each of our decision,
and pick the highest. Thus
if we choose 'H' (honest), the expected profit $\hprof$ is
\[ \hprof = \val * \pmargin * (1-\score) + \score * (- \val) * d \]
That is,
with probability $\score$ we allow a fraudulent order,
and we lose its value ($- \val * d $);
and with probability (1-$\score$) we correctly allow the order,
for a profit of $\val * \pmargin$.
Similarly, the expected profit $\fprof$ when we decide 'fraud' is
\[ \fprof = -\invCost - \friction * (1-\score) ~ + ~  \score * (-\invCost)   \]
Thus, we should decide to 'pass' (ie.,  'honest'), if
\[ 
\hprof > \fprof
\]
Now we solve for the inequality, where expanding out both sides we have:
\begin{eqnarray}
\nonumber
-[(1-s)(\invCost+f)+s*\invCost] &>& (1-s) * \val*\pmargin-s*\val*d \nonumber 
\end{eqnarray}
Solving for $s$, we have:
\begin{eqnarray}
\frac{(\invCost+f+\val*\pmargin)}{(\val*d+\val*\pmargin+ f)} & > & s\nonumber 
\end{eqnarray}
And with one final rearrangement
\begin{eqnarray}
s < 1 - \frac{\val  * d  - \invCost}{\val  * (\pmargin + d) + \friction} = \shMax \nonumber 
\end{eqnarray}
we arrive at equation~\ref{eq:decisionHonest}.
 \hfill $\blacksquare$